\documentclass[sigconf]{acmart}

\usepackage{bm}
\usepackage{bbm}
\usepackage{color}
\usepackage{CJKutf8}
\usepackage{algorithmic}
\usepackage{algorithm}
\usepackage{bbding} 
\usepackage{multirow}
\usepackage[flushleft]{threeparttable}
\usepackage{flushend}
\usepackage{amsmath}
\usepackage{balance}
\usepackage{multirow}
\usepackage{booktabs}
\usepackage{enumerate}
\usepackage{amsfonts}
\usepackage{threeparttable}
\usepackage[misc]{ifsym}
\usepackage{diagbox}
\usepackage{enumitem}
\usepackage{indentfirst}

\newcommand{\ignore}[1]{}
\newcommand\ChangeRT[1]{\noalign{\hrule height #1}}

\usepackage{etoolbox}

\newtheorem{theorem}{Theorem}
\usepackage{bbm}
\usepackage{apptools}
\AtAppendix{\counterwithin{theorem}{section}}
\usepackage{graphicx}
\copyrightyear{2022}

\acmYear{2022}
\setcopyright{acmcopyright}\acmConference[CIKM '22]{Proceedings of the 31st ACM International Conference on Information and Knowledge Management}{October 17--21, 2022}{Atlanta, GA, USA}
\acmBooktitle{Proceedings of the 31st ACM International Conference on Information and Knowledge Management (CIKM '22), October 17--21, 2022, Atlanta, GA, USA}
\acmPrice{15.00}
\acmDOI{10.1145/3511808.3557145}
\acmISBN{978-1-4503-9236-5/22/10}

\begin{document}
\title{Approximated Doubly Robust Search Relevance Estimation}
\author{Lixin Zou$^\dagger$, Changying Hao$^\dagger$, Hengyi Cai, Suqi Cheng, \\ Shuaiqiang Wang, Wenwen Ye, Zhicong Cheng, Simiu Gu, Dawei Yin{\small $^{\textrm{\Letter}}$}}
\affiliation{
\{zoulixin15, cyhaocn, hengyi1995, chengsuqi, shqiang.wang\}@gmail.com,\\ \{chengzhicong01, gusimiu\}@baidu.com, yindawei@acm.org \\
Baidu Inc., China}
\thanks{
$^\dagger$ Fair contribution. \\
$^{\textrm{\Letter}}$ Corresponding author. 
}
\renewcommand{\shortauthors}{Zou, et al.}

\begin{abstract}
Extracting query-document relevance from the sparse, biased clickthrough log is among the most fundamental tasks in the web search system. Prior art mainly learns a relevance judgment model with semantic features of the query and document and ignores directly counterfactual relevance evaluation from the clicking log. Though the learned semantic matching models can provide relevance signals for tail queries as long as the semantic feature is available. However, such a paradigm lacks the capability to introspectively adjust the biased relevance estimation whenever it conflicts with massive implicit user feedback. The counterfactual evaluation methods, on the contrary, ensure unbiased relevance estimation with sufficient click information. However, they suffer from the sparse or even missing clicks caused by the long-tailed query distribution. 

In this paper, we propose to unify the counterfactual evaluating and learning approaches for unbiased relevance estimation on search queries with various popularities. Specifically, we theoretically develop a doubly robust estimator with low bias and variance, which intentionally combines the benefits of existing relevance evaluating and learning approaches. We further instantiate the proposed unbiased relevance estimation framework in Baidu search, with comprehensive practical solutions designed regarding the data pipeline for click behavior tracking and online relevance estimation with an approximated deep neural network. Finally, we present extensive empirical evaluations to verify the effectiveness of our proposed framework, finding that it is robust in practice and manages to improve online ranking performance substantially.




\end{abstract}

%
%


\maketitle

\section{Introduction}
Nowadays, search engines play an ever more crucial role in meeting users' information needs by locating relevant web pages from a prohibitively large corpus. Query-document relevance estimation, the core task in search result ranking, has been the most critical problem since the birth of web search. Numerous works have been proposed to extract the relevance signals given the query and a large corpus of documents~\citep{10.1145/2939672.2939677}, including direct text matching~\citep{10.1561/1500000019,manning2008introduction} and link analysis~\citep{Pageetal98}.

In the industrial setting, implicit feedback (e.g., clicks) usually acts as an attractive data source to conduct query-document relevance estimation, which exhibits detailed and valuable information about users' interactions and can be collected at virtually no additional costs.
Moreover, unlike explicit relevance judgments obtained from experts or crowd-sourcing, which can become obsolete quickly, implicit feedback reflects the time-varying preferences of the actual user population and is easy to maintain.

Unfortunately, data collected from user interactions in web search does not necessarily capture the true utility of each document for each query, owning to the presence of intrinsic click bias in user interactions.
For example, \textit{position bias} occurs since users are less likely to examine, and thus click, lower-ranked items~\citep{craswell2008experimental,luo2022model}, and \textit{trust bias} arises because users trust the ranking system and are more likely to click on documents at higher ranks that are not relevant~\citep{10.1145/3308558.3313697}. 
As a result, inferring the query-document relevance from user interactions becomes particularly challenging, especially when user-issued queries and clicked documents follow the long-tail distribution.

Recent work on unbiased relevance estimation with implicit user feedback can be broadly categorized into two groups. The first group counterfactually learns an unbiased relevance judgment model with the semantic feature~(e.g., TF-IDF, BM25~\citep{robertson2009probabilistic}). 
Typically, it achieves the goal through reweighting the ranking loss by treating examination as a counterfactual effect~\citep{DBLP:conf/sigir/AiBLGC18,DBLP:conf/www/HuWPL19}. 
It ensures the distribution-insensitive performance as long as the text content of the query and document is available, which is particularly beneficial for tail queries/documents. 
Remarkably, the neural relevance estimators fine-tuned from pre-trained language models (PLMs) using click behavioral data~\citep{10.1145/3447548.3467147}, establish the state-of-the-art ranking effectiveness, attributing to its superior generalization ability.
However, solely relying on the semantic matching model for relevance estimation can lead to sub-optimal results, since it inherently can not adjust the biased search results using only the query/document text, despite the fact that massive user behavioral data can be leveraged to de-bias search results effectively.
The second group, counterfactual relevance evaluation, focuses on extracting unbiased and reliable relevance signals from biased click feedback.
One such effort is the click model~\citep{chuklin2015click}, which models the click with relevance and bias factors and extracts the relevance with parameter estimation techniques~\citep{DBLP:conf/www/ChapelleZ09,mao2018constructing,wang2013incorporating}.
For example, \citet{DBLP:conf/www/ChapelleZ09} propose a straightforward approach to re-weight the clicks with the inverse bias-dependent expected clicks, i.e., {\it the clicks over expected clicks}. 
However, for most tail queries, the click information is too sparse or missing entirely, which hinders the effectiveness of directly inferring the relevance from logging data~\citep{mao2019investigating}. 
In this work, we contend that it is beneficial to unify the counterfactual evaluation  and learning methods for unbiased relevance estimation on search queries with various popularities.
To make up the deficiencies of existing relevance estimation approaches,
this work analyzes the properties of each fundamental relevance estimator and devises a novel way to intentionally combine their strengths, yielding a doubly robust relevance estimation framework.
Specifically, we begin with a theoretical investigation of the Inverse Propensity Weighting (IPW)-based approach and analyze their capability to correct biased feedback regarding the bias and variance of the estimator.
Building on these insights, we introduce a PLM-based imputation model, working as a distribution-robust relevance estimator. 
By drawing on doubly robust estimation techniques~\citep{dudik2011doubly,pmlr-v97-wang19n,10.1145/3366423.3380037,kang2007ddr,saito2020doubly,yuan2020unbiased,DBLP:conf/sigir/GuoZLYCWCY021}, we then develop a provably unbiased estimator, combining benefits of the low-variance PLM-based imputation model and the low-bias IPW-based estimator, for relevance estimation using biased feedback data. 
We further instantiate the proposed unbiased relevance estimation framework in Baidu's commercial search engine, with comprehensive practical solutions designed regarding the data flow for click feature tracking and online inference by approximating the doubly robust relevance estimator with a deep neural network.
Finally, we present extensive empirical evaluations to verify the effectiveness of our proposed framework, finding that it is robust in practice and manages to improve online ranking performance substantially.
Our main contributions can be summarized as follows:
\begin{itemize}[leftmargin=*]
    \item We propose a doubly robust estimator for unbiased relevance estimation in web search, and theoretically analyze the bias and variance of the estimator.
    \item Based on the doubly robust estimator, we design an effective and sensible system workflow to field the proposed framework in a large-scale ranking system.
    \item We conduct extensive offline and online experiments to validate the effectiveness of the proposed framework. The results show that the proposed techniques significantly boost the search engine's performance.
\end{itemize}

\section{Preliminary}
In this section, we first present the user behavior assumption used in this paper and formulate the task of counterfactual relevance estimation. We then review the conventional method for unbiased relevance estimation and discuss its strength and weakness.
\subsection{User Behavior Model}
Suppose we have $K$ positions in a typical search engine, and for a query $q$, the document $d$ is displayed at position $k \in [1, K]$. 
Let $C$ be a binary random variable indicating whether the user clicks document $d$, $E$ whether the user examines the document $d$, and $R$ for the true relevance.

Following the generative \textbf{position-based model}~(PBM)~\cite{10.1145/3308558.3313697}, users click the document if and only if they examine the document and the document is relevant. 
Additionally, the examination probability only depends on the position $k$, but not on $q$ or $d$. 
Based on these assumptions, we have the following probability for a click
\begin{eqnarray*}
Pr(C=1|q,d,k) = Pr(E=1|k) \cdot Pr(R=1|q,d)
 = \theta_k \cdot \gamma_{q,d},
\end{eqnarray*}
where we use $\theta_k$ and $\gamma_{q,d}$ as short-hands for the corresponding probability. 

The PBM approach has a noise-free assumption that the perceived relevance, denoted as $\tilde{R}$, is the same as the true relevance.
However, this is not the case in real-world scenarios. Existing user eye-tracking study~\cite{DBLP:conf/sigir/JoachimsGPHG05} shows that users are more likely to trust, and thus click, the highly ranked results due to their trust in the effectiveness of the search engine to rank relevant documents higher.
As a result, a relevant document can be missed, and meanwhile a non-relevant document can be clicked,
which reveals that a position-dependent \textbf{trust bias} in addition to the examination bias (position bias) is present in click data.
In this paper, we model the trust bias as follows,
\begin{eqnarray*}
Pr(\tilde{R}|R=1,E=1,k) &= \epsilon^{+}, \\
Pr(\tilde{R}|R=0,E=1,k) &= \epsilon^{-},
\end{eqnarray*}
which indicates that a relevant document at position $k$ can be missed with probability 1-$\epsilon^{+}$, and a non-relevant document can be clicked mistakenly with probability $\epsilon^{-}$.
With the trust bias being considered, the clicking probability is reformulated as
\begin{eqnarray}
\nonumber
P(C=1 \mid d,q, k)&=&\theta_k\left(\epsilon_{k}^{+} \gamma_{q,d}+\epsilon_{k}^{-}\big(1 -\gamma_{q,d}\big)\right) \\ \label{equ_trust_bias}
                &=& \theta_k(\epsilon_{k}^{+} - \epsilon_{k}^{-})\gamma_{q,d} + \theta_k\epsilon_{k}^{-}.
\end{eqnarray}
For simplifying the notation, we denote $\alpha_k = \theta_k(\epsilon_{k}^{+} - \epsilon_{k}^{-})$ and $\beta_k = \theta_k\epsilon_{k}^{-}$, resulting in a compact notation for the click probability $P(C=1 \mid d,q, k) = \alpha_k\gamma_{q,d} + \beta_k$.

\subsection{Counterfactual Relevance Estimation}
Let $\mathcal{D}_{q,d} = \{(k_{i},c_{i})\}^D_{i=1}$ be a set of collected interaction data of a query-document pair $(q,d)$ over a period (e.g., the last day, last week or last month) with $c_{i}\in\{0,1\}$ indicating the observed clicking behavior. 
$D$ denotes the number of observed clicking data. 
Counterfactual relevance estimation concerns the approaches that estimate the relevance from the historical interactions as
\begin{eqnarray*}
\hat{\gamma}_{q,d} =\frac{1}{D}\sum_{(k,c)\in \mathcal{D}_{q,d}}\vartheta(k,c).
\end{eqnarray*}
Here, $\vartheta$ is a counterfactual estimator, aiming to minimize the difference between $\hat{\gamma}_{q,d}$ and $\gamma_{q,d}$ as 
\begin{eqnarray*}
\ell(\hat{\gamma}_{q,d}) &=& \mathbb{E}_{\mathcal{D}_{q,d}}\left[(\hat{\gamma}_{q,d} - \gamma_{q,d})^2\right]\\
&=& \left\{\text{Bias}_{\mathcal{D}_{q,d}}\left[\hat{\gamma}_{q,d}\right]\right\}^2 + \text{Var}_{\mathcal{D}_{q,d}}\left[\hat{\gamma}_{q,d}\right] + \sigma^2,
\end{eqnarray*}
where $\text{Bias}_{\mathcal{D}_{q,d}}\left[\hat{\gamma}_{q,d}\right]$ and
$\text{Var}_{\mathcal{D}_{q,d}}\left[\hat{\gamma}_{q,d}\right]$ are the bias and variance of $\hat{\gamma}_{q,d}$, and $\sigma^2$ is an irreducible error~\citep{domingos2000unified}. 
Therefore, the goal of counterfactual relevance estimation is equivalent to introducing the $\hat{\gamma}_{q,d}$ that minimizes its \textbf{bias} and \textbf{variance}.
\section{Relevance Evaluation with Inverse Propensity Weighting}\label{sec_ipw}
A straightforward approach is to estimate $\gamma_{q,d}$ by re-weighting the click according to the examination probability $\theta_k$~\cite{DBLP:conf/www/ChapelleZ09}, leading to an IPW estimator as
\begin{eqnarray*}
\hat{\gamma}^{IPW}_{q,d} =\frac{1}{D}\sum_{(k,c)\in \mathcal{D}_{q,d}} \frac{c}{\hat{\theta}_{k}},
\end{eqnarray*}
where $\hat{\theta}_k$ is the estimated examination probability given the position $k$. 
With a properly specified $\hat{\theta}_k = \theta_k$, the $\hat{\gamma}^{IPW}_{q,d}$ would be an unbiased estimator under the position-biased model.

However, ~\citet{10.1145/3308558.3313697} prove that $\hat{\gamma}^{IPW}_{q,d}$ cannot correct for the trust bias.
Alternatively, they introduce an estimator based on affine corrections. This {\it affine} estimator penalizes an item displayed at rank $k$ by $\hat{\beta}_k$ while also re-weights the clicks inversely w.r.t. $\hat{\alpha}_k$ as
\begin{eqnarray*}
\hat{\gamma}^{aff}_{q,d} = \frac{1}{D}\sum_{(k,c)\in \mathcal{D}_{q,d}}\frac{1}{\hat{\alpha}_{k}} \left(c-\hat{\beta}_{k}\right),
\end{eqnarray*}
where $\hat{\alpha}_{k}$ and $\hat{\beta}_{k}$ are the estimated $\alpha_{k}$ and $\beta_{k}$ respectively.  

\textbf{Bias and Variance Analysis.}
We formulate the bias and variance of the affine estimator to analyze its strength and weakness. 
\begin{theorem}\label{th_ips_bias_variance}
Let $\Delta_{\alpha_k}$ and $\Delta_{\beta_k}$ be the simplified notation of $(\alpha_{k} - \hat{\alpha}_{k})$ and $(\beta_{k} - \hat{\beta}_{k})$ respectively.
Then the bias and variance of the $\hat{\gamma}^{aff}_{q,d}$ estimator are
\begin{eqnarray}
    \label{ipw_bias}
    \nonumber
            \text{Bias}_{\mathcal{D}_{q,d}}\left[\hat{\gamma}^{aff}_{q,d}\right] &=& \frac{1}{D}\sum_{(k,c)\in\mathcal{D}_{q,d}}\frac{\Delta_{\alpha_k}\gamma_{q,d} +\Delta_{\beta_k}}{\hat{\alpha}_{k}}, \\
    \label{ipw_var}
            \text{Var}_{\mathcal{D}_{q,d}}\left[\hat{\gamma}^{aff}_{q,d}\right] &=& \frac{1}{D}\sum_{(k,c)\in\mathcal{D}_{q,d}} \frac{\left(\hat{\alpha}_k\hat{\gamma}^{aff}_{q,d} + \hat{\beta}_k - c\right)^2}{\hat{\alpha}^2_k}.
    \end{eqnarray} 
\end{theorem}
The full derivation is presented in Appendix~\ref{appendix_theory_1}. 

\paragraph{\bf Strength and Weakness of the Affine Estimator} 
In Theorem~\ref{th_ips_bias_variance}, the affine estimator is equal to the ideal relevance if $\Delta_{\alpha_k}=0$ and $\Delta_{\beta_k}=0\, \forall k\in [1,K]$. 
However, it has following limitations in practice: 
\textbf{(1)} the unbiasedness of affine estimator is hard to obtain since the $\alpha_k$ and $\beta_k$ are unknown in reality. 
Existing approaches estimate it with sophisticated computations, such as the MLE, EM algorithm~\citep{chuklin2015click,10.1145/3308558.3313697},
while accurately estimating its true value remains a formidable challenge;
\textbf{(2)} the affine estimator is inaccurate for tail queries. As presented in Equation~\ref{ipw_var}, the variance of affine estimator is negatively correlated with the size of $\mathcal{D}_{q,d}$, indicating an unneglectable estimation error when it comes to the tail queries with sparse, or even missing clicking data.

\section{Learning Relevance with PLM-based Model}\label{sec_plm}
The inability to generalize the click information to the tail queries is the main barrier for improving relevance estimation using clicking data.
An alternative for relevance estimation relies on the semantic matching between the query and document, which is able to capture the query-document relevance from the linguistic perspective and enjoys superior generalizations on tail queries.
With the recent significant progress of pre-training language models (PLMs) like BERT~\cite{devlin2018bert,hao2021sketch} and ERNIE~\cite{sun2020ernie} in many language understanding tasks, large-scale pre-trained models also demonstrate increasingly promising text ranking results, and neural rankers fine-tuned from pre-trained language models establish state-of-the-art ranking effectiveness~\cite{ma2021prop,chu2022h}. 

\textbf{Relevance Modeling with PLM.}
Following previous work~\cite{10.1145/3447548.3467147}, we employ an ERNIE 2.0\footnote{https://github.com/PaddlePaddle/ERNIE} based cross-encoder~\cite{10.1145/3447548.3467147} to capture the query-document relevance, in which deeply-contextualized representations of all possible input token pairs bridge the semantic gap between query and document terms.
Specifically, we first concatenate the tokens of the query (q-tokens) and document (d-tokens) into a token sequence and then feed it into the ERNIE encoder to obtain the high-dimension representation as
\begin{equation} \label{encoder}
    \bm{h}_{\text{cls}} =\text{ERNIE([CLS]} \circ \text{q-tokens} \circ[\mathrm{SEP}] \circ \text{d-tokens} \circ[\mathrm{SEP}]),
\end{equation}
where $\circ$ denotes the concatenate operation, $\bm{h}_{\text{cls}}\in\mathbb{R}^{768}$ is the last layer's ``[CLS]'' token representation, and ``[SEP]'' denotes the special symbol to separate non-consecutive token sequences.
Particularly, a 12-layer encoder is employed for extracting the dense semantic representation. 
Given $\bm{h}_{\text{cls}}$, a fully-connected network is used as the scoring module to predict the imputed relevance score
\begin{align*} 
    \hat{\gamma}_{q,d}^{imp} &= \text{sigmoid}(\bm{w}^\top \cdot \bm{h}_{\text{cls}} + b),
\end{align*}
where $\text{sigmoid}(x) = {1}/({1+\exp({-x})})$ is the activation function. 
$\bm{w}\in\mathbb{R}^{768}$ and $b\in\mathbb{R}$ are trainable weight and bias parameter. 

\textbf{Fine-tuning with Randomization Data.}
We fine-tune the PLM-based relevance estimator with clean and high-quality data. 
More concretely, we collect unbiased data $\mathcal{D}_{rand} = \{(q_{i},d_{i},c_{i})\}_{i=1}^N$, containing $N$ tuples of query, document and user's implicit feedback, by randomly presenting the query-document pair $(q,d)$ to online users in \textbf{the first ranking place and only the implicit feedback on top-1 document is recorded for training}. The concern of position-related biases is therefore eliminated since the top-1 result is most likely examined by users and all the documents are equally ranked in the first position. 
The PLM-based imputation model is then fine-tuned with $\mathcal{D}_{rand}$, by minimizing the following cross-entropy loss
\begin{eqnarray*}
\ell(\hat{\gamma}^{imp}) =  \sum_{(q,d,c)\in\mathcal{D}_{rand}}-c\log \hat{\gamma}_{q,d}^{imp}-(1-c)\log \left(1-\hat{\gamma}_{q,d}^{imp}\right).
\end{eqnarray*}

\textbf{\bf Strength and Weakness of the Imputation Model.} 
Given the text of query and document, the PLM-based imputation estimator is capable of providing relevance signals, even for those queries with sparse or even missing clicks. 
However, for the high-frequency queries, the PLM-based imputation model can not update its predicted relevance even when its estimated relevance violates with majority users' behaviors.
In other words, the imputation model is a low-variance but biased estimator that lacks the capability of introspectively being consistent with credible user feedback.


\section{Unified Relevance Modeling with Doubly Robust Estimation}
Either evaluating the search relevance with the aforementioned affine estimator or learning it with the PLM-based imputation model can be considered band-aid solution: the former strategy excels at relevance estimation when abundant user clicks can be collected, however, it risks high-variance estimations for tail queries; whilst learning relevance merely with the PLM-based model is also suboptimal since the system can not introspectively adjust its predicted relevance even when the model results are conflict with the massive credible user feedback. This observation motivates us to develop a unified relevance estimation model, which can organically integrates the existing two schemes to collectively fulfill the entire goal. 
By exploiting the introduced PLM-based semantic matching model as an imputation model, we frame the PLM-based imputation estimator and the affine estimator into the doubly robust relevance estimation schema~\citep{dudik2011doubly,DBLP:conf/sigir/GuoZLYCWCY021}, taking the best of both worlds. 

Specifically, by leveraging a PLM-based imputation model, we estimate the ideal relevance between the query and document with $\gamma^{imp}_{q,d}$ estimator as
\begin{eqnarray} 
\nonumber
\hat{\gamma}_{q,d}^{dr} &=& \hat{\gamma}_{q,d}^{imp} + \frac{1}{D} \sum_{(k,c)\in\mathcal{D}_{q,d}} \frac{c-\hat{\beta}_{k} - \hat{e}
_{k} (\hat{\epsilon}^+_{k} -\hat{\epsilon}^-_{k})\hat{\gamma}_{q,d}^{imp}}{\hat{\alpha}_{k}}  \\
\label{equ_dr_final}
&=& \frac{1}{D} \sum_{(k,c)\in\mathcal{D}_{q,d}} \frac{\hat{\alpha}_{k}-\hat{e}
_{k} (\hat{\epsilon}^+_{k} - \hat{\epsilon}^-_{k})}{\hat{\alpha}_{k}}\hat{\gamma}_{q,d}^{imp} + \hat{\gamma}_{q,d}^{aff}.
\end{eqnarray}


Notably, we further incorporate an estimated click-correlated examination indicator $\hat{e}_{k}$, which can be approximated with user's behaviors, such as the click and dwelling time on the search result. \textbf{The detail is presented in \S~\ref{sec_exam}.}

\textbf{Bias and Variance Analysis.} 
We analyze the bias and variance of the doubly robust estimator to demonstrate its improvement over the existing approaches.

\begin{theorem}\label{th_dr_bias_variance}
Let $\Tilde{\alpha}_k$ be the simplified notation $\hat{e}_{k}(\hat{\epsilon}^+_{k} - \hat{\epsilon}^-_{k})$, $\Delta_{\Tilde{\alpha}_{k}}$ be short of $\Tilde{\alpha}_k - \hat{\alpha}_{k}$.
The bias and variance of the $\hat{\gamma}^{dr}_{q,d}$ estimator are
{\small
\begin{eqnarray*}
\text{Bias}_{\mathcal{D}_{q,d}}\left[\hat{\gamma}^{dr}_{q,d}\right] &=&\frac{1}{D} \sum_{(k,c)\in\mathcal{D}_{q,d}} \frac{\Delta_{\alpha_{k}}\gamma_{q,d}+ \Delta_{\beta_{k}}- \Delta_{\Tilde{\alpha}_{k}}\hat{\gamma}^{imp}_{q,d}}{\hat{\alpha}_{k}} \\
\text{Var}_{\mathcal{D}_{q,d}}\left[\hat{\gamma}^{dr}_{q,d}\right] &=& \frac{1}{D}\sum_{(k,c)\in\mathcal{D}_{q,d}}  \frac{(\Tilde{\alpha}_{k}\hat{\gamma}^{imp}_{q,d} + \hat{\beta}_{k} - c)^2}{\hat{\alpha}^2_{k}} + \delta_k,
\end{eqnarray*} 
}
where $\delta_k = \frac{\left(\hat{\gamma}^{dr}_{q,d} - \hat{\gamma}^{imp}_{q,d}\right)\left(\hat{\alpha}_{k}\hat{\gamma}^{dr}_{q,d} + (2\Tilde{\alpha}_{k}- \hat{\alpha}_{k})\hat{\gamma}^{imp}_{q,d} + 2\hat{\beta}_{k} - 2c\right)}{\hat{\alpha}_{k}} $. 
\end{theorem}

The full derivation of Theorem~\ref{th_dr_bias_variance} is presented in Appendix~\ref{appendix_theory_2}. 
If $\mathbb{E}_{\mathcal{D}_{q,d}}[\Tilde{\alpha}_k - \alpha_k]=0 $, the bias term can be rearranged as $\text{Bias}_{\mathcal{D}_{q,d}}\left[\hat{\gamma}^{dr}_{q,d}\right]=\frac{1}{D} \sum_{(k,c)\in\mathcal{D}_{q,d}} \frac{\Delta_{\alpha_{k}}(\gamma_{q,d}-\hat{\gamma}^{imp}_{q,d})+ \Delta_{\beta_{k}}}{\hat{\alpha}_{k}}$.
We see that if either $\Delta_{\alpha_{k}}=0$ or $(\gamma_{q,d}-\hat{\gamma}^{imp}_{q,d}) = 0$, the unbiasedness of $\hat{\gamma}^{dr}_{q,d}$ is guaranteed with $\Delta_{\beta_k}=0$. This property is called doubly robustness.

For the variance term, we empirically have $\delta_k<0$ with properly specified $\hat{e}_{k}$, $\hat{\alpha}_{k}$ and $\hat{\beta}_{k}$. Taking a single displayed query-document pair as an example. If $c=0$, then $(\hat{\gamma}^{dr}_{q,d} - \hat{\gamma}^{imp}_{q,d}) <0$ and $(\hat{\alpha}_{k}\hat{\gamma}^{dr}_{q,d} + (2\Tilde{\alpha}_{k}- \hat{\alpha}_{k})\hat{\gamma}^{imp}_{q,d} + 2\hat{\beta}_{k} - 2c) = (\Tilde{\alpha}_{k}\hat{\gamma}^{imp}_{q,d} +\hat{\beta}_{k} ) > 0$. 
Otherwise, if $c=1$, we can infer that  $(\hat{\gamma}^{dr}_{q,d} - \hat{\gamma}^{imp}_{q,d}) > 0$ and $(\hat{\alpha}_{k}\hat{\gamma}^{dr}_{q,d} + (2\Tilde{\alpha}_{k}- \hat{\alpha}_{k})\hat{\gamma}^{imp}_{q,d} + 2\hat{\beta}_{k} - 2c) = $ $(\Tilde{\alpha}_{k}\hat{\gamma}^{imp}_{q,d} +\hat{\beta}_{k} - 1)< 0$. Furthermore, with accurately estimated relevance $\gamma_{q,d}$, $(\Tilde{\alpha}_{k}\gamma_{q,d} + \hat{\beta}_{k} - c)^2$ is typically smaller than $(\hat{\alpha}_k\gamma_{q,d} + \hat{\beta}_k - c)^2$ in Equation~\ref{ipw_var} since $\hat{e}_{k}$ in $\Tilde{\alpha}_{k}$ is positively correlated with $c$ to some extent, refered to \S~\ref{sec_exam}.

\textbf{Advantage of the Doubly Robust Estimator.}
By inspecting the derived bias and variance, we contend that the advantage of the doubly robust estimator lies in three aspects:
\textbf{(1)} With correctly specified $\alpha_k, \beta_k$ or $\hat{\gamma}_{q,d}^{imp}$, the unbiased doubly robust estimator is acquired, demonstrating its robustness regarding the high-frequency queries. 
\textbf{(2)} The doubly robust estimator reduces the variance by incorporating an imputation model, which ensures its effectiveness on the tail queries. \textbf{(3)} The doubly robust estimator is flexible in updating the imputation model with any deep semantic matching model, eliminating the updating burden for re-training the PLM-based model.



\section{Approximated Relevance Estimation in Online System}
\begin{figure}[t]
\includegraphics[width=1.0\linewidth]{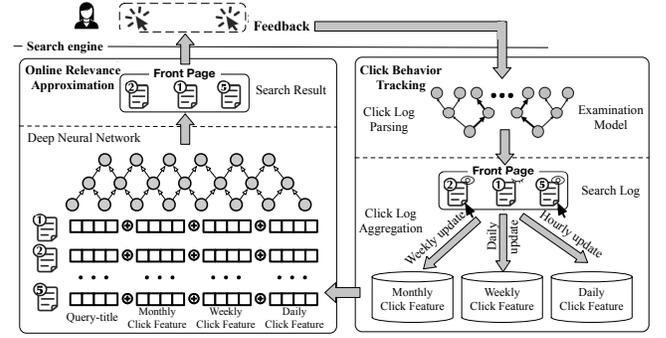}
\caption{
The framework of the deployed doubly robust relevance estimator.
}
\label{fig_deployment_framework}
\end{figure}
We instantiate our designed doubly robust relevance estimator in a production environment. Figure~\ref{fig_deployment_framework} depicts the overall workflows of the deployed framework. 
It consists of two parts: 
\textbf{(1)} \textbf{Click Behavior Tracking}, labeling examinations of every search result, and collecting the essential click information, e.g., display position, click, examination for every query-document pair;
\textbf{(2)} \textbf{Online Relevance Approximation}, calculating the ranking score with a deep neural network that approximates the doubly robust relevance estimator.

\subsection{Click Behavior Tracking}
Click behavior tracking involves an examination model to track user examination behaviors and update the click features for online inference. 

\begin{figure}[h]
\includegraphics[width=1.0\linewidth]{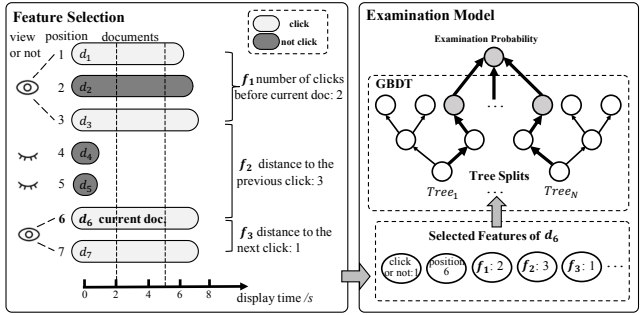}
\caption{
The proposed examination model that estimates the examination based on the user's behavior log.
}
\label{exam_model}
\end{figure}

\subsubsection{Examination Model}\label{sec_exam}
Faithfully perceiving user examination behaviors is critical for the success of doubly robust relevance estimator. 
However, it is hard to obtain unless tracking the user's eyeballs. 
In this work, as shown in Figure~\ref{exam_model}, we propose an examination model which estimates the examination based on the user's behavior log. 
The used features, training data, and the predictive model are specified as follows: 
\textbf{(1)~Feature Selection:} 
We extract a series of informative features for examination estimation, such as the ranking position,
being clicked or not, number of clicks before the current document, distance to the previous click, etc. These essential features can be leveraged to deduce user's examination behavior.
\textbf{(2) Training Data Mining:}  
To train the examination model, we mine the search log data and utilize the display time on the screen of terminal device $t_k$ as an indicator.  
The longer the document is left on the main page, the higher the probability it will be seen by users. Further, a user will never check the document without being displayed on the screen. 
Here, we set examples with $t_k > 5s$ as the positive samples and $t_k < 1s $ as the negatives.
Typically, $5s$ is the average time spent before clicking the top-1 result.
\textbf{(3) Examination Model:} 
Labeling user examinations is an hourly-updated task over all the logging data
since tracking the time-varying preferences of the actual user population is crucial for the counterfactual relevance estimation.
To reconcile its effectiveness and efficiency,
we employ a gradient boosting decision tree (GBDT)~\cite{ke2017lightgbm} as the backbone of the examination model. 
Given the examination label, we train the examination model with a cross-entropy loss as
\begin{eqnarray*}
\ell_{\hat{e}_k} = -\mathbbm{1}[t_{k}>5]\log \hat{e}_k -\mathbbm{1}[t_{k}<1] \log (1-\hat{e}_k),
\end{eqnarray*}
where $\mathbbm{1}$ denotes the indicator function. 

\subsubsection{Click Behavior Tracking} 
We next introduce the feature tracking system, which records click features used in online inference. 
It is implemented with multiple regularly running map-reduce jobs on Baidu's distributed computing platform. 
As shown in Figure~\ref{fig_deployment_framework}, the offline feature tracking system consists of two stages: 
\textbf{1) Click Log Parsing}. 
In this stage, we parse the logging data of every search and obtain the click, display, displaying time, dwelling time, and page skipping information at every ranking position. 
Given the display information, the examination model is utilized to estimate the examination probability of every query-document pair. 
\textbf{2) Click Log Aggregation}. 
In the second stage, the historical users' behaviors are merged into three dictionaries that record the clicking behaviors over the month, week, and day, resulting with a dense clicking vector $\bm{x}\in \mathbb{R}^L$ containing $L$ click features. 
Particularly, the monthly dict updates every weekend and records the last four weeks' clicking behaviors. The weekly dict updates every day and records the last week's clicking behaviors. 
The daily dict updates every hour and records the last 24 hours clicking behaviors.


\subsection{Online Relevance Approximation}
To efficiently serve the proposed framework in an online search engine system,
we propose to approximate the doubly robust estimator with a neural network architecture, as shown in Figure~\ref{fig_dr_estimated}. 
Specifically, following the Equation~\ref{equ_dr_final}, we decouple the doubly robust estimator into three parts
\begin{eqnarray*}
\hat{\gamma}_{q,d}^{dr} = \zeta_{q,d}\hat{\gamma}_{q,d}^{imp} + \hat{\gamma}_{q,d}^{aff},
\end{eqnarray*}
where $\zeta_{q,d} = \frac{1}{D} \sum_{(k,c)\in\mathcal{D}_{q,d}} \frac{\hat{\alpha}_{k}-\hat{e}
_{k} (\hat{\epsilon}^+_{k} - \hat{\epsilon}^-_{k})}{\hat{\alpha}_{k}}$ 
is a data-dependent trade-off coefficient that balances the estimation of $\hat{\gamma}_{q,d}^{imp}$ and $\hat{\gamma}_{q,d}^{aff}$. 
For $\hat{\gamma}_{q,d}^{imp}$, it is implemented as an ERNIE-based model, as described in \S~\ref{sec_plm}. 
The $\hat{\gamma}_{q,d}^{aff}$ is approximated with a neural network by aligning the historical clicking behavior with the unbiased randomization clicking data. 
Finally, the trade-off coefficient $\zeta_{q,d}$ is estimated through balancing the $\hat{\gamma}_{q,d}^{imp}$ and $\hat{\gamma}_{q,d}^{aff}$ on the unbiased randomization clicking data. 

\begin{figure}
\includegraphics[width=1.0\linewidth]{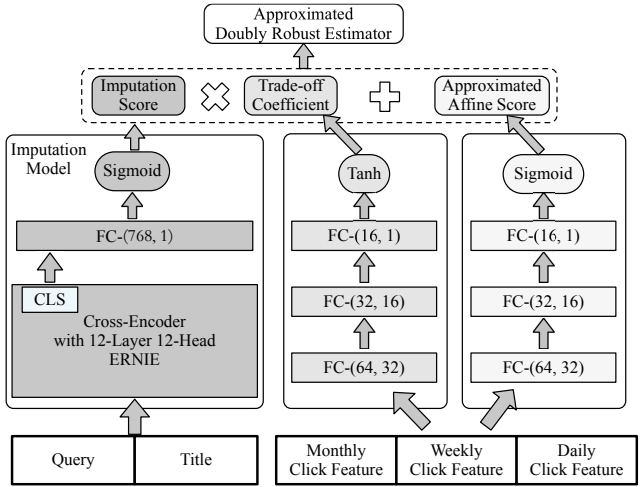}
\caption{
The approximated doubly robust estimator that consists of three parts: the imputation model, the trade-off coefficient and the approximated affine estimator. 
}
\label{fig_dr_estimated}
\end{figure}

\subsubsection{Approximated Affine Estimator}\label{sec_aprroxi_affine}
The affine estimator takes the recorded biased clicking log $\bm{x}_{q,d}$ as input and outputs an unbiased relevance estimation with hyper-parameters $\hat{\alpha}_{k}$ and $\hat{\beta}_{k}$. 
To this end, we directly align $\bm{x}_{q,d}$ to its relevance by training a 3-layer multilayer perceptron~(MLP) with the randomization data~$\mathcal{D}_{rand}$ as
\begin{eqnarray*}
\ell(\bar{\gamma}^{aff}) &=& \sum_{(q,d,c)\in\mathcal{D}_{rand}}-c\log \bar{\gamma}_{q,d}^{aff}-(1-c)\log \left(1-\bar{\gamma}_{q,d}^{aff}\right) \\
&s.t.& \bar{\gamma}_{q,d}^{aff} = \text{sigmiod}(\text{MLP}(\bm{x}_{q,d})),
\end{eqnarray*}
where $\bar{\gamma}^{aff}_{q,d}$ is the approximated affine estimator with the neural network, $\bm{x}_{q,d}\in\mathbb{R}^L$ is the click feature for query-document pair $(q,d)$.

\subsubsection{Trade-off Coefficient}
We blend the $\bar{\gamma}^{aff}_{q,d}$ and $\hat{\gamma}^{imp}_{q,d}$ with a trade-off coefficient, which incorporates the $\hat{\gamma}^{imp}_{q,d}$ by considering the historical logging data as
\begin{eqnarray*}
\bar{\gamma}_{q,d}^{dr} &=& \bar{\zeta}_{q,d}(\perp\hat{\gamma}_{q,d}^{imp}) + (\perp\bar{\gamma}_{q,d}^{aff}) \\
&s.t.& \bar{\zeta}_{q,d} = \tanh(\text{MLP}(\bm{x}_{q,d})).
\end{eqnarray*}
Here, $\perp$ is an operator that sets the gradient of the operand to zero\footnote{This operator $\perp$ has been implemented in scientific computing libraries, e.g., {\it stop gradient} in TensorFlow and {\it detach} in PyTorch.}.
We approximate its value with a MLP and adopt the $\tanh(\cdot)$ as the activation function. Then, we optimize with the policy gradient algorithm~\citep{Williams:92} using unbiased data $\mathcal{D}_{rand}$ as
\begin{eqnarray}\label{equ_policy_gradient}
\ell(\bar{\zeta}) =   \frac{1}{|\mathcal{D}_{rand}|}\sum_{(q,d,c)\in\mathcal{D}_{rand}} - \hat{c} \log{\bar{\gamma}_{q,d}^{dr}},
\end{eqnarray}
where $\hat{c} = 2c-1$ converts the click into a reward for conducting optimization.








\section{Experiments}
To assess the effectiveness of the proposed solutions, we conduct extensive offline and online experiments on a large-scale real-world search system. 
This section details the experimental setup and presents several insights demonstrating that the proposed approaches are effective for online ranking in a commercial search engine system.

\subsection{Dataset}
We train and evaluate our proposed method with both logged user behavioral data~(\textbf{log}), randomization data~(\textbf{rand}) and manually-labeled data~(\textbf{manual}). The \textbf{log} and \textbf{rand} are collected for training model, and \textbf{manual} is adopted as the evaluation set. Specifically, we randomly sample billions of query-document pairs from a week of users' accessing logs from Sep. 2021. The \textbf{rand} data is collected from Sep. 2021 to Nov. 2021, which consists of 10 millions query-document pairs. 
In the manually-labeled data, $400,000$ query-document pairs for $8,000$ queries are annotated for off-line evaluation. Table~\ref{tab:data_statistics} offers the dataset statistics.

\subsection{Evaluation Methodology}
We employ the following evaluation metrics to assess the performance of the ranking system.


The \textbf{Discounted Cumulative Gain} (DCG)~\citep{Jrvelin2017IREM} is a standard listwise accuracy metric and is widely adopted in the context of ad-hoc retrieval. 
For a ranked list of $N$ documents, we use the following implementation of DCG
\begin{eqnarray*}
    DCG_N = \sum_{i=1}^N \frac{G_i}{\log_2(i+1)},
\end{eqnarray*}
where $G_i$ represents the weight assigned to the document's label at position $i$.
Higher degree of relevance corresponds to a higher weight.
We use the symbol $DCG$ to indicate the average value of this metric over the test queries. 
$DCG$ will be reported only when absolute relevance judgments are available. 
In online experiments, we extract $6,000$ queries and manually label the top-4 ranking results generated by the search engine for calculating $DCG$.  

The \textbf{Expected Reciprocal Rank}~(ERR)~\cite{ghanbari2019err} considers the importance of the document at a position to be dependent on the documents ranked higher than this document. This measure is defined as
\begin{eqnarray*}
    ERR_{N}=\sum_{i=1}^{N} \frac{1}{i} \prod_{j=1}^{i-1}\left(1-R_j\right) R_i,
\end{eqnarray*}
where $R_i$ indicates the relevance probability of the $i$-th document to the query and the expression $\frac{1}{i} \prod_{j=1}^{i-1}\left(1-R_j\right)$ represents the non-relevance probability of the ordered documents prior to the position of the $i$-th document in the list.


The \textbf{Good vs. Same vs. Bad} (GSB)~\citep{10.1145/3447548.3467147} is a metric measured by the professional annotators' judgment. 
For a user-issued query, the annotators are provided with a pair (result$_1$, result$_2$) in which one result is returned by system A, and the other is generated by a competitor system B. 
The annotators, who do not know which system the result is from, are then required to independently rate among Good (result$_1$ is better), Bad (result$_2$ is better), and Same (they are equally good or bad), considering the relevance between the returned document and the given query. 
In order to quantify the human evaluation, we aggregate these three indicators mentioned above as a unified metric, denoted as $\Delta \text{GSB}$
\begin{eqnarray*}
    \Delta \text{GSB} = \frac{\# \text{Good} - \# \text{Bad}}{\# \text{Good} + \# \text{Same} +  \# \text{Bad}}.
\end{eqnarray*}

\begin{table}
\centering
\caption{Data statistics. K: thousand; M: million; B: billion.}
\scalebox{1.0}{
\begin{tabular}{l|rr}
\ChangeRT{0.8pt}
Data & $\#$Query & $\#$Query-Document Pairs \\ \hline
\textbf{log} data     & 1B & 4B  \\
\textbf{rand} data    &  6M   & 10M \\
\textbf{manual} data  & 8K & 400K \\
\ChangeRT{0.8pt}
\end{tabular}
}
\label{tab:data_statistics}
\end{table}

{\centering
\small
\tabcolsep 0.02in
\renewcommand{\arraystretch}{1.0}
\begin{table*}[h]
\centering
\caption{The relative DCG and ERR improvements over the baseline approaches.}
\label{tab_offline_result}
    \centering
    \scalebox{0.92}{
    \begin{tabular}{ c l r r r r r r r r r r}
    \ChangeRT{0.8pt}
 & & \multicolumn{5}{c}{DCG@K } & \multicolumn{5}{c}{ERR@K }  \\
    \cmidrule(lr){3-7}  
    \cmidrule(lr){8-12} 
     & & K = 2 & K = 4 & K = 6 & K = 10 & K = 20 & K = 2 & K = 4 & K = 6 & K = 10 & K = 20 \\
    \ChangeRT{0.5pt} 
\multirow{5}{*}{Cross-Encoder} 
     & No-correction (Base) & - & - & - & - & - & - & - & - & - & -\\
     & PAL                  & +9.852\%	 &+7.194\%	 &+5.583\%	 &+3.695\%	 &+0.703\%	 &+9.322\%	 &+7.317\%	 &+6.169\%	 &+5.573\%	 &+5.438\% \\  
     & TAL                  &+11.355\%	 &+8.429\%	 &+6.750\%	 &+4.784\%	 &+1.336\%	 &+11.017\%	 &+8.711\%	 &+7.468\%	 &+6.811\%	 &+6.647\%	 \\
     & IPW  &+11.872\%	 &+8.327\%	 &+6.361\%	 &+4.304\%	 &+0.744\%		 &+11.441\%	 &+8.711\%	 &+7.143\%	 &+6.811\%	 &+6.344\% \\              
    & Display Filtering & +10.144\%	 &+8.313\%	 &+7.127\%	 &+5.590\%	 &+3.124\%		 &+9.746\%	 &+8.014\%	 &+7.143\%	 &+6.502\%	 &+6.344\% \\
     & Rand             &+12.298\%	 &+10.362\%	 &+9.404\%	 &+8.256\%	 &+6.219\%	 &+11.864\%	 &+10.105\%	 &+8.766\%	 &+8.359\%	 &+7.855\% \\ \ChangeRT{0.1pt}
 \multirow{5}{*}{Mixture Model} 
     & Vanilla-Mix + No-correction &+21.342\%	 &+9.795\%	 &+3.398\%	 &-3.215\%	 &-9.396\%		 &+22.458\%	 &+15.331\%	 &+12.338\%	 &+11.146\%	 &+10.574\%	\\
     & Vanilla-Mix + Display Filtering &+24.731\%	 &+16.945\%	 &+12.687\%	 &+8.162\%	 &+2.767\%		 &+\textbf{24.576\%}	 &+19.164\%	 &+16.883\%	 &+15.480\%	 &+14.804\%	 \\
     & Vanilla-Mix + Rand  &+19.232\%	 &+15.259\%	 &+13.122\%	 &+11.128\%	 &+8.188\%	 &+18.644\%	 &+15.331\%	 &+13.636\%	 &+12.693\%	 &+12.085\% \\
     & Approximated Affine Estimator &+16.136\%	 &+11.742\%	 &+8.992\%	 &+6.430\%	 &+5.621\%		 &+16.102\%	 &+12.892\%	 &+11.039\%	 &+10.217\%	 &+9.970\%	\\ 
     & Approximated Doubly Robust & +\textbf{24.776\%}$^{\uparrow}$	 &+\textbf{19.183\%}$^\uparrow$	 &+\textbf{16.634\%}$^\uparrow$	 &+\textbf{14.146\%}$^\uparrow$	 &+\textbf{11.770\%}$^\uparrow$	 &+\textbf{24.576\%}$^\uparrow$	 &+\textbf{20.209\%}$^\uparrow$	 &+\textbf{18.182\%}$^\uparrow$	 &+\textbf{16.718\%}$^\uparrow$	 &+\textbf{16.314\%}$^\uparrow$	 \\ 
    \ChangeRT{0.8pt}
    \multicolumn{12}{c}{ ``${\uparrow}$'' indicates a statistically significant improvement ($t$-test with $p<0.05$ over the baselines).} 
    \end{tabular}
    }
\end{table*}}

\subsection{Competitor System}
According to whether involving the clicking features, we investigate two kinds of backbone models for all methods:
\textbf{(1) Cross-Encoder},
which models the query-document relevance with a 12-layer cross-encoder architecture, as mentioned in \S~\ref{sec_plm}. 
It obtains superior performance for text matching tasks~\cite{humeau2019poly}. 
\textbf{(2) Mixture Model},
which incorporates the user behavior data for document ranking. 
Particularly, we study three kinds of its variants: 
\begin{itemize}[leftmargin=*]
    \item \textbf{Vanilla-Mix}: it first concatenates the representation of cross-encoder $\bm{h}_{cls}$ and clicking feature $\bm{x}$, and then predicts the relevance with a 3-layer MLP.
    \item \textbf{Approximated Affine Estimator}: it solely relies on the clicking feature $\bm{x}$ for relevance estimation, as mentioned in \S~\ref{sec_aprroxi_affine}.
    \item \textbf{Approximated Doubly Robust}: the proposed method in this paper. 
\end{itemize}

Grounding on the aforementioned backbone models, we study the influence of different debiasing approaches:
\begin{itemize}[leftmargin=*]
    \item \textbf{No-correction}: it uses the raw click data to train the ranking model. Its performance can be considered as a lower bound for the unbiased ranking model.
    \item \textbf{PAL}: the position-bias aware learning framework~\cite{DBLP:conf/recsys/GuoYLTZ19} introduces a position model to explicitly represent position bias, and jointly models the bias and relevance. 
    \item \textbf{TAL}: we further equips the PAL with trust-bias modeling as in Equation~\ref{equ_trust_bias}. TAL is the short for \textbf{T}rust-bias \textbf{A}ware \textbf{L}earning.
    \item \textbf{IPW}: it weights the learning-to-rank loss with the propensity score mentioned in \S~\ref{sec_ipw}. The propensity score is estimated using the randomization data.
    \item \textbf{Display Filtering}: it trains the model with the query-document pairs that have been displayed on the screen, i.e., the displaying time $t>0s$.
    \item \textbf{Rand}: it trains the model with the randomization click data on the top-1 ranking position, which is described in \S~\ref{sec_plm}.
\end{itemize}

\subsection{Experimental Setting}
Regarding the backbone cross-encoder model, we use a 12-layer transformer architecture.
It is warm-initialized with a 12-layer ERNIE 2.0 provided by the Baidu Wenxin\footnote{https://wenxin.baidu.com/} toolkit. The inference speed of 12-layer ERNIE is optimized with TensorRT\footnote{https://developer.nvidia.com/tensorrt} and model quantization~\cite{gholami2021survey}, which \textbf{save more than 90\% inference time}.
The same hyper-parameters are used for various comparison models, i.e., vocabulary size of $18,000$, hidden size of $768$, feed-forward layers with dimension $1024$, and batch size of $128$.
For model optimization, we use the Adam~\citep{Kingma2015AdamAM} optimizer with a dynamic learning rate following~\citet{Vaswani2017AttentionIA}. 
We set the warm-up steps as 4000 and the maximum learning rate as $2\times 10^{-6}$ both in the pre-training and fine-tuning stage. Without particular specification, the MLP layers are implemented with the hidden size of \{64, 32, 16\}. 
All the models are trained on the distributed computing platform with $28$ Intel(R) 5117 CPU, $32G$ Memory, 8 NVIDIA V100 GPUs, and 12T Disk. 
\subsection{Offline Evaluation}
Table~\ref{tab_offline_result} summarizes different methods' relative DCG and ERR improvements over the baseline approaches. We observe that:
\begin{itemize}[leftmargin=*]
\item Our proposed method outperforms all the state-of-the-art approaches in terms of DCG and ERR, indicating that our proposed model is effective and robust to the real-world clicking behaviors.
\item The approximated affine estimator is an effective unbiased relevance estimator, which is advantageous on all metrics over the No-correction baseline. 
\item Addressing trust bias is advantageous for debiasing the clicking data since TAL is much better than PAL when handling the trust bias. It further beats the IPW on all metrics except top-2 metrics. \item The recorded historical clicking features help to improve the performance. When comparing the performance of Cross-Encoder with Mixture Model, we notice that the performance of the Mixture Model are substantially better than the corresponding method trained with the Cross-Encoder architecture. 
\item Filtering the data with displaying time $t>0s$ is beneficial for dealing with the examination bias. The baseline approach with display filtering strategy outperforms the baselines without corrections, e.g., No-correction and Vanilla-Mix. 
\item Using randomization data for model training is an effective way for correcting the bias. As shown in Table~\ref{tab_offline_result}, the performance of fine-tuning with randomization data is much better than training with uncorrected data. Furthermore, it brings more performance improvements after top-10 results since the randomized displaying increases the possibility of exposing more items beyond the top-10 results. 

\end{itemize}

\subsection{Online Experiments}
To investigate the effectiveness of the introduced techniques in the real-world commercial search engine, we deploy the proposed model to a online search system and compare its performance with the most competitive base model in the real production environment.
The online \textbf{base} model is a vanilla mixture model that conducts semantic matching between the query and document with convolution neural network and bag-of-words~(CNN+Bow-Mix) model~\cite{shen2014learning}. From the Table~\ref{online_ab_improve}, we have the following observations:
\begin{itemize}[leftmargin=*]
    \item Upgrading the semantic matching model from CNN+Bow to a cross-encoder architecture substantially improves the overall performance, especially for the tail queries. As shown in Table~\ref{online_ab_improve}, the improvements on long-tail queries regarding the DCG and GSB are +1.270\% and +3.500\%. It is quite consistent with our intuition that the PLM-based semantic matching model is more beneficial for the tail queries.
    \item The proposed approximated doubly robust estimator achieves best performance on the relevance, especially for the tail queries. We observe that the proposed approach beats the online base system by a large margin with +1.46\% and +8.5\% relative improvements on DCG@4 and GSB for long-tail queries respectively.
\end{itemize}
\begin{table}[h]
\centering
\small
\renewcommand{\arraystretch}{1}
\tabcolsep 0.02in
\caption{Performance improvements of online A/B testing.}
\scalebox{0.95}{
\begin{tabular}{l|cc|cc}
\ChangeRT{1pt}
\multirow{2}{*}{Model} & \multicolumn{2}{c|}{DCG@4} & \multicolumn{2}{c}{$\Delta \text{GSB}$} \\ 
& Random & Long-Tail & Random & Long-Tail \\\hline
CNN+Bow-Mix+Display Filtering &  -            &  -    &  -         &  -      \\ 
Vanilla-Mix+Display Filtering &  +0.420\%$^{\uparrow}$ & +1.270\%$^{\uparrow}$  &  +0.500\% &   +3.500\%$^{\uparrow}$ \\
Approximated Doubly Robust &  \textbf{+0.710}\%$^{\uparrow}$  &  \textbf{+1.460}\%$^{\uparrow}$ & \textbf{+1.600}\%$^{\uparrow}$   &  \textbf{+8.500}\%$^{\uparrow}$\\
\ChangeRT{1pt}
\multicolumn{5}{c}{``${\uparrow}$'' indicates a statistically significant improvement}\\ 
\multicolumn{5}{c}{ ($t$-test with $p<0.05$ over the baselines).}
\end{tabular}
}
\label{online_ab_improve}
\end{table}
\subsection{Effectiveness of the Trade-off Coefficient}
To investigate the impact of the learned trade-off coefficient, we further compare the performance of evaluation method~(i.e., Approximated Affine Estimator), learning method~(i.e., PLM-based Imputation Model) and its combination~(i.e., Approximated Doubly Robust) versus different search frequency in Table~\ref{tab:co_efficent}. From the table, we observe that:
\begin{itemize}[leftmargin=*]
    \item The trade-off coefficient can efficiently balance the trade-off between evaluation and learning methods. As shown in Table~\ref{tab:co_efficent}, the approximated doubly robust estimator outperforms these methods in terms of DCG@4 by a large margin. 
    \item The evaluation-based affine estimator is fragile when ranking documents with tail queries. Nevertheless, it exhibits superior performance for high-frequency and middle-frequency search queries. This observation is quite consistent with our analysis in \S~\ref{sec_ipw}. 
    \item The imputation model surpasses the evaluation method by a large margin~(+23.34\%) on the tail queries but underperforms on high-frequency and middle-frequency setting, which supports our analysis in \S~\ref{sec_plm}.
\end{itemize}


\subsection{The Effectiveness of Examination Model}
The examination model plays a vital role in the proposed doubly robust relevance estimation framework. 
We investigate the effectiveness of the designed examination model by analyzing its accuracy on a labeled test set that contains ten thousands of data instances.
Unsurprisingly, it reaches a \textbf{0.95} AUC score~\cite{huang2005using} on the test set, which reveals that the proposed examination model is able to accurately predict users' examination behaviors. 
Additionally, as shown in Figure~\ref{fig_exam_probability}, we plot \textbf{(a)} the average estimated examination probability at different positions and \textbf{(b)} the estimated examination probability on the subsequent positions below the anchor click using the devised examination model.
From the figure, we notice that the examination probability is decayed with the ranking position but has a slight boost at the last position on the first page, i.e., position 10. It is reasonable since users get used to pull-down to the last search result if their information needs are not well met by the system. Moreover, as expected, the examination probability significantly drops on the subsequent positions down from the anchor click, implying that our proposed examination model is able to effectively deduce user’s examination behaviors.
\begin{figure}[h]
\setlength{\fboxrule}{0.pt}
\setlength{\fboxsep}{0.pt}
\fbox{
\includegraphics[width=0.9\linewidth]{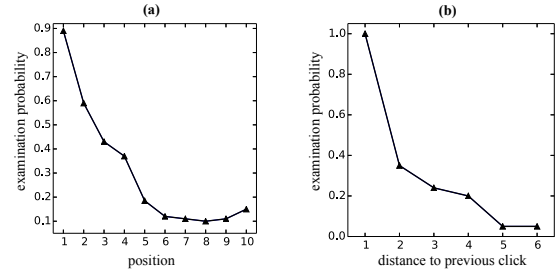}
}
\caption{
(a) The estimated average examination probability of different positions. (b) The estimated average examination probability of the subjacent ranking positions.}
\label{fig_exam_probability}
\end{figure}
\begin{table}[h]
\centering
\small
\caption{Performance comparison of evaluation and learning methods versus different search frequency. High: \#Query Search/\#Month $\geq$ 60,000, Mid: 10 $\leq$  \#Query Search/\#Month < 60,000 ; Tail:  \#Query Search/\#Month < 10.}
\scalebox{1.0}{
\begin{tabular}{l|ccc}
\ChangeRT{1pt}
\multirow{2}{*}{Model} & \multicolumn{3}{c}{DCG@4}  \\ 
& High & Mid & Tail \\ \hline
Approximated Affine Estimator & - & - & - \\ 
PLM-based Imputation Model & -3.10\% & -1.49\% & +23.34\%$^{\uparrow}$\\ 
Approximated Doubly Robust & \textbf{+4.37}\%$^{\uparrow}$ & \textbf{+4.63}\%$^{\uparrow}$ & \textbf{27.46}\%$^{\uparrow}$ \\\ChangeRT{1pt}
\multicolumn{4}{c}{``${\uparrow}$'' indicates a statistically significant improvement}\\ 
\multicolumn{4}{c}{ ($t$-test with $p<0.05$ over the baselines).}\\ 
\end{tabular}
}
\label{tab:co_efficent}
\end{table}
















\section{Related Work}
User behavior data such as clicks has been shown to be quite promising in improving ranking models' performance~\cite{yang2022can,zou2020neural,zou2020pseudo,zou2019reinforcement,luo2022model,zhao2020whole,lin2021disentangled}. However, directly treating the click as relevance judgment might lead to misleading evaluation results or sub-optimal ranking functions due to various types of bias in clicks, e.g., position bias~\cite{craswell2008experimental}, and trust bias~\cite{10.1145/3308558.3313697}. In the past, a large amount of research has been devoted to extracting accurate and reliable relevance signals from clicking data. Depending on whether training an unbiased ranking model, existing methods can be broadly categorized into two research lines: 

\textbf{Counterfactual evaluation methods} attempt to directly extract unbiased and reliable relevance signals from biased click feedback. 
For instance, ~\citet{joachims2002optimizing} treat clicks as relevance between clicked and skipped documents. 
~\citet{richardson2007predicting} assume that a user would only click a document when the user observed it and considered it relevant and further propose an examination hypothesis to model the position bias.
Consequently, a series of click models have been proposed to model the examination probability and infer accurate relevance feedback from user clicks~\cite{chapelle2009expected,chuklin2015click,mao2018constructing,wang2013incorporating}. 
In ~\citet{craswell2008experimental}, a cascade click model, modeling user’s sequential browsing behavior on the search engine, has been proposed.
\citet{dupret2008user} propose a user browsing model that allows users to read with jumps from previous results to later results. 
Nonetheless, despite their differences, click models usually require that the same query-document pair appear multiple times for reliable relevance inference~\cite{mao2019investigating}, hindering their effectiveness for tail search queries.

\textbf{Counterfactual learning to rank} methods attempt to learn a ranking model with biased user feedback so that the resultant ranking model will converge to the same model trained with unbiased relevance labels. 
Inverse propensity weighting~(IPW), reweighting the ranking loss with an inverse propensity score, is a widely used approach. 
Early endeavors estimate the propensity score by result randomization~\cite{DBLP:conf/wsdm/JoachimsSS17,DBLP:conf/sigir/WangBMN16}. 
Recent efforts jointly model the propensity estimation and unbiased learning to rank. 
For instance, \citet{DBLP:conf/sigir/AiBLGC18} and \citet{DBLP:conf/www/HuWPL19} present a dual learning framework for estimating the click bias and training the ranking model. 
Further work extends IPW to various biases, such as trust bias~\cite{10.1145/3308558.3313697} and context bias~\cite{wu2021unbiased}. 
Derived from direct bias modeling methods, researchers also contribute a series of approaches which focus on modeling bias factors into click models and extracting true relevance from click signals by removing the bias factor in the inference stage~\cite{DBLP:conf/www/ChapelleZ09, craswell2008experimental, dupret2008user,DBLP:conf/recsys/GuoYLTZ19}. 
By using a ranking model that involves various semantic features, counterfactual learning methods is able to conduct relevance estimation, even for those queries with sparse or even missing clicks as long as the text of query and document is available. 
However, such a paradigm lacks the capability to introspectively adjust the biased relevance estimation whenever it conflicts with massive implicit user feedback. 

To make up for the deficiencies of existing relevance evaluation and learning methods, this work analyzes the properties of each fundamental relevance estimator and devises a novel way to intentionally combine their strengths, yielding a doubly robust relevance estimation framework.

\section{Conclusions}

This paper presents a novel unbiased relevance estimation method using the clickthrough data recorded by web search engine logs.
Existing works on this topic suffer from sparse or even missing click signals.
The model we propose bypasses this deficiency by introducing a PLM-based semantic imputation model and arranging it in the doubly robust relevance estimation framework. It leads to an unbiased relevance estimation approach taking the best of both counterfactual relevance estimators with user clicks and neural relevance estimators fine-tuned from pre-trained language models.
Moreover, we introduce a series of practical techniques in click behavior tracking and convenient online relevance approximation. 
The resulting model is effective and highly applicable in a production environment, which has been deployed and tested at scale in Baidu's commercial search engine.

\bibliographystyle{ACM-Reference-Format}
\appendix
\section{Appendix}
\subsection{Bias and Variance Analysis of IPW Estimator}\label{appendix_theory_1} 
\begin{theorem}
Let $\Delta_{\alpha_k}$ and $\Delta_{\beta_k}$ be the simplified notation of $(\alpha_{k} - \hat{\alpha}_{k})$ and $(\beta_{k} - \hat{\beta}_{k})$ respectively.
Then the bias and variance of the $\hat{\gamma}^{aff}_{q,d}$ estimator are
    {\small\begin{eqnarray*}
            \text{Bias}_{\mathcal{D}_{q,d}}\left|\hat{\gamma}^{aff}_{q,d}\right] &=& \frac{1}{D}\sum_{(k,c)\in\mathcal{D}_{q,d}}\left[\frac{\Delta_{\alpha_k}\gamma_{q,d} +\Delta_{\beta_k}}{\hat{\alpha}_{k}}\right], \\
            \text{Var}_{\mathcal{D}_{q,d}}\left[\hat{\gamma}^{aff}_{q,d}\right] &=& \frac{1}{D}\sum_{(k,c)\in\mathcal{D}_{q,d}} \frac{\left(\hat{\gamma}^{aff}_{q,d}\hat{\alpha}_k + \hat{\beta}_k - c\right)^2}{\hat{\alpha}^2_k}.
    \end{eqnarray*}} 
\end{theorem}

\begin{proof}
{\small \begin{eqnarray*}
\text{Bias}_{\mathcal{D}_{q,d}}\left[\hat{\gamma}^{aff}_{q,d}\right] &=& \mathbb{E}_{\mathcal{D}_{q,d}}\left[\frac{c - \hat{\beta}_{k}}{\hat{\alpha}_{k}}\right] - \gamma_{q,d}\\
&=& \mathbb{E}_{\mathcal{D}_{q,d}}\left[\frac{\alpha_{k}  \gamma_{q,d} + \beta_{k} - \hat{\beta}_{k}}{\hat{\alpha}_{k}}\right] - \gamma_{q,d}\\
&=& \frac{1}{D}\sum_{(k,c)\in\mathcal{D}_{q,d}}\left[\frac{\Delta_{\alpha_k}\gamma_{q,d} +\Delta_{\beta_k}}{\hat{\alpha}_{k}}\right]
\end{eqnarray*}
\begin{eqnarray*}
\text{Var}_{\mathcal{D}_{q,d}}\left[\hat{\gamma}^{aff}_{q,d}\right]
&=& \frac{1}{D}\sum_{(k,c)\in\mathcal{D}_{q,d}} \left(\hat{\gamma}^{aff}_{q,d} - \frac{c-\hat{\beta}_k}{\hat{\alpha}_k} \right)^2\\
&=& \frac{1}{D}\sum_{(k,c)\in\mathcal{D}_{q,d}} \frac{(\hat{\gamma}^{aff}_{q,d}\hat{\alpha}_k + \hat{\beta}_k - c)^2}{\hat{\alpha}^2_k} 
\end{eqnarray*}}
\end{proof}

\subsection{Bias and Variance Analysis of Doubly Robust Estimator}
\label{appendix_theory_2} 

\begin{theorem}\label{th_dr_bias_variance_appendix}
Let $\Tilde{\alpha}_k$ be the simplified notation $\hat{e}_{k}(\hat{\epsilon}^+_{k} - \hat{\epsilon}^-_{k})$, $\Delta_{\Tilde{\alpha}_{k}}$ be short of $\Tilde{\alpha}_k - \hat{\alpha}_{k}$.
The bias and variance of the $\hat{\gamma}^{dr}_{q,d}$ estimator are
{\small
\begin{eqnarray*}
\text{Bias}_{\mathcal{D}_{q,d}}\left[\hat{\gamma}^{dr}_{q,d}\right] &=&\frac{1}{D} \sum_{(k,c)\in\mathcal{D}_{q,d}}\left[ \frac{\Delta_{\alpha_{k}}\gamma_{q,d}+ \Delta_{\beta_{k}}- \Delta_{\Tilde{\alpha}_{k}}\hat{\gamma}^{imp}_{q,d}}{\hat{\alpha}_{k}}\right] \\
\text{Var}_{\mathcal{D}_{q,d}}\left[\hat{\gamma}^{dr}_{q,d}\right] &=& \frac{1}{D}\sum_{(k,c)\in\mathcal{D}_{q,d}}  \frac{(\Tilde{\alpha}_{k}\hat{\gamma}^{imp}_{q,d} + \hat{\beta}_{k} - c)^2}{\hat{\alpha}^2_{k}} + \delta_k,
\end{eqnarray*} 
}
where $\delta_k = \frac{(\hat{\gamma}^{dr}_{q,d} - \hat{\gamma}^{imp}_{q,d})(\hat{\alpha}_{k}\hat{\gamma}^{dr}_{q,d} + (2\Tilde{\alpha}_{k}- \hat{\alpha}_{k})\hat{\gamma}^{imp}_{q,d} + 2\hat{\beta}_{k} - 2c)}{\hat{\alpha}_{k}} $. 
\end{theorem}

\begin{proof}
{\small\begin{eqnarray*}
&&\text{Bias}_{\mathcal{D}_{q,d}}\left[\hat{\gamma}^{dr}_{q,d}\right] \\
&=& \mathbb{E}_{\mathcal{D}_{q,d}}\left[\frac{c-\hat{\beta}_{k} - \Tilde{\alpha}_{k}\hat{\gamma}^{imp}_{q,d}}{\hat{\alpha}_{k}}\right] + \hat{\gamma}_{q,d}^{imp} - \gamma_{q,d} \\
&=& \mathbb{E}_{\mathcal{D}_{q,d}}\left[\frac{\alpha_{k}\gamma_{q,d} + \beta_{k}-\hat{\beta}_{k} - \Tilde{\alpha}_{k}\hat{\gamma}^{imp}_{q,d}}{\hat{\alpha}_{k}}\right] + \hat{\gamma}_{q,d}^{imp} - \gamma_{q,d} \\
&=&\frac{1}{D} \sum_{(k, c) \in \mathcal{D}_{q, d}}\left[\frac{\Delta_{\alpha_{k}} \gamma_{q, d}+\Delta_{\beta_{k}}-\Delta_{\Tilde{\alpha}_{k}} \hat{\gamma}_{q, d}^{i m p}}{\hat{\alpha}_{k}}\right]
\end{eqnarray*}

\begin{eqnarray*}
&&\text{Var}_{\mathcal{D}_{q,d}}\left[\hat{\gamma}^{dr}_{q,d}\right] \\
&=& \text{Var}_{\mathcal{D}_{q,d}}\left[\frac{c-\hat{\beta}_{k} - \hat{e}_{k}(\hat{\epsilon}^+_{k} - \hat{\epsilon}^-_{k})\hat{\gamma}_{q,d}^{imp}}{\hat{\alpha}_{k}}\right] \\ 
&=& \frac{1}{D}\sum_{(k,c)\in\mathcal{D}_{q,d}} \frac{\left(\hat{\alpha}_k(\hat{\gamma}^{dr}_{q,d} - \hat{\gamma}^{imp}_{q,d}) + (\Tilde{\alpha}_{k}\hat{\gamma}^{imp}_{q,d} + \hat{\beta}_{k} - c)  \right)^2}{\hat{\alpha}^2_k}\\
&=& \frac{1}{D}\sum_{(k,c)\in\mathcal{D}_{q,d}}  \frac{(\Tilde{\alpha}_{k}\hat{\gamma}^{imp}_{q,d} + \hat{\beta}_{k} - c)^2}{\hat{\alpha}^2_{k}} + \\
&&\frac{(\hat{\gamma}^{dr}_{q,d} - \hat{\gamma}^{imp}_{q,d})(\hat{\alpha}_{k}\hat{\gamma}^{dr}_{q,d} + (2\Tilde{\alpha}_{k}- \hat{\alpha}_{k})\hat{\gamma}^{imp}_{q,d} + 2\hat{\beta}_{k} - 2c)}{\hat{\alpha}_{k}} 
\end{eqnarray*}}
\end{proof}


\newpage
\bibliography{reference}

\end{document}